\documentclass[preprint,12pt]{elsarticle}

\usepackage{amssymb}
\usepackage{amsmath}
\usepackage{amsthm}

\newtheorem{theorem}{Theorem}
\newtheorem{lemma}{Lemma}
\newtheorem{corollary}{Corollary}
\theoremstyle{definition}
\newtheorem{definition}{Definition}
\newcommand{\idc}{\mathit{id}_c}

\journal{Discrete Applied Mathematics}

\begin{document}

\begin{frontmatter}

\title{Detection and Construction of Small Blocking Sets via Critical Independent Sets and Matching Decomposition}

\author[inst1]{Quanlong Niu}

\author[inst2]{Riheng Jia\corref{cor1}}
\ead{rihengjia@zjnu.edu.cn}

\cortext[cor1]{Corresponding author}

\affiliation[inst1]{
    organization={School of Big Data and Artificial Intelligence, Fuyang University of Technology},
    city={Fuyang},
    postcode={236000},
    state={Anhui},
    country={China}
}

\affiliation[inst2]{
    organization={School of Computer Science and Technology, Zhejiang Normal University},
    city={Jinhua},
    postcode={321004},
    state={Zhejiang},
    country={China}
}

\begin{abstract}
The small blocking set detection and construction problem asks whether a given connected simple graph contains a nonempty independent set whose cardinality equals the size of its neighborhood and is strictly less than half the number of vertices, and, if so, how to construct such a set. Previous work characterized restless explorability in vertex-permuted temporal graphs using small blocking sets, but left the computational complexity of detecting their existence as an open problem. We give a polynomial-time algorithm for detection and construction. Specifically, we first compute a maximum-cardinality critical independent set and use its difference and cardinality to handle the cases that admit an immediate solution. In the remaining boundary case, where the critical difference is zero and the selected independent set contains exactly half the vertices, we establish a projection theorem using a perfect matching and neighborhood counting: the graph contains a small blocking set if and only if the selected independent set contains a nonempty proper tight subset. We then use matching contraction to translate the tight-subset condition into the absence of arcs leaving a vertex subset in an auxiliary digraph. This yields a strong-connectivity test for the existence of a small blocking set and, when one exists, a construction based on a sink strongly connected component. The result resolves the open complexity question for small blocking set detection and provides an explicit construction algorithm.
\end{abstract}

\begin{keyword}
Small blocking set \sep Critical independent set \sep Projection theorem \sep Perfect matching \sep Strongly connected component
\end{keyword}

\end{frontmatter}

\section{Introduction}
\label{sec:introduction}

Connections in communication, transportation, and social contact networks are inherently temporal, and the order of interactions affects reachability and propagation. Dynamic and temporal graphs are therefore widely used to model and analyze such systems~\cite{holme2012,casteigts2012}. For example, Pan and Saram\"aki used temporal networks to study paths and centrality in human communication and air transportation, revealing differences between static distances and temporal reachability~\cite{pan2011}. Salath\'e et al. recorded interpersonal contacts using wireless sensors to construct a high-resolution contact network for studying infectious disease transmission~\cite{salathe2010}.

Research on computation and communication under dynamic topology has addressed fundamental problems such as information dissemination, distributed computation, and network exploration. Kuhn, Lynch, and Oshman studied distributed computation in continuously changing networks and analyzed how connectivity and stability affect computational efficiency~\cite{kuhn2010}. In temporal graph exploration, a mobile agent must visit all vertices using edges available at the corresponding time steps. Michail and Spirakis investigated algorithms and complexity for this problem~\cite{michail2016}. Erlebach, Hoffmann, and Kammer further studied exploration time bounds, hardness of approximation, and exploration algorithms for special classes of underlying graphs~\cite{erlebach2021}, while Erlebach and Spooner obtained improved exploration time bounds for degree-bounded temporal graphs~\cite{erlebach2018}.

Restless exploration additionally requires the agent to move at every time step, rather than wait for an edge to become available. For infinite periodic temporal graphs, Bellitto, Conchon-Kerjan, and Escoffier showed that the problem is solvable in polynomial time for period two and is NP-hard for every fixed period of at least three~\cite{bellitto2023}. These results demonstrate that the combination of temporal variation and restrictions on waiting gives exploration a structure and complexity distinct from static graph traversal.

Ayoubi and Narayanan studied restless exploration and token dissemination in vertex-permuted temporal graphs~\cite{ayoubi2025}. Every snapshot in such a graph is isomorphic to the same base graph, but the correspondence between vertices may change over time. They introduced blocking sets and proved that, for sufficiently long lifetimes, every vertex-permuted temporal graph based on a connected graph admits restless exploration if and only if the base graph contains no blocking set of cardinality strictly less than half its number of vertices. They left the complexity of detecting such a set as an open problem. Direct enumeration of vertex subsets can test this condition, but the number of candidates grows exponentially in the worst case.

We give a polynomial-time algorithm for detecting and constructing small blocking sets. We first compute a maximum-cardinality critical independent set using an existing algorithm~\cite{larson2007} and perform an initial classification according to its difference and cardinality. For the remaining boundary case, where the critical difference is zero and the selected set contains exactly half the vertices, we establish an equivalence between the existence of a small blocking set in the graph and the existence of a nonempty proper tight subset within the selected independent set. This reduces the search over the entire graph to a structural test within a fixed independent set. We then use a perfect matching to construct a digraph, translate tightness into the absence of arcs leaving a vertex subset, and compute the condensation digraph using a strongly connected component algorithm~\cite{tarjan1972}. A vertex of outdegree zero in the condensation digraph corresponds to a sink strongly connected component (sink SCC), whose corresponding vertices on the independent-set side form a tight subset. If this set contains strictly fewer than half the vertices of the original graph, it is the required small blocking set. If the auxiliary digraph is strongly connected, its only component corresponds to the entire independent set and does not satisfy the strict cardinality condition. Together, these steps give a polynomial-time detection and construction algorithm.

\section{Problem formulation}
\label{sec:formulation}

Let $G=(V,E)$ be a finite, connected, undirected simple graph, with $n=|V|$ and $m=|E|$, and assume $n\geq 2$. Ayoubi and Narayanan studied restless exploration in vertex-permuted temporal graphs (VPGs). In this model, the topology of the base graph $G$ remains unchanged, but vertex identities may be permuted at every time step; thus, every snapshot is isomorphic to $G$. Given a sequence of permutations $\Pi=(\pi_1,\ldots,\pi_L)$ of $V$, the corresponding temporal graph is
\[
\mathcal{G}^{\Pi}=(G_1,\ldots,G_L),
\]
where
\[
G_t=(V,E_t),\qquad
E_t=\bigl\{\{\pi_t(u),\pi_t(v)\}:\{u,v\}\in E\bigr\}.
\]
An agent starts at a specified vertex and must traverse an available edge at each time step, without waiting at a vertex. The objective of restless exploration is to find a restless temporal walk that visits every vertex. Combining Theorems 10 and 14 of Ayoubi and Narayanan, for lifetime $L\geq n(n-1)$, every VPG based on $G$ admits restless exploration from every starting vertex if and only if $G$ contains no small blocking set~\cite{ayoubi2025}. Equivalently, if $G$ contains a small blocking set, one can construct a permutation sequence and a starting vertex for which the corresponding VPG admits no restless exploration.

\begin{definition}[Small blocking set]
\label{def:small}
For an independent set $S\subseteq V$, let $N(S)$ denote the set of vertices adjacent to at least one vertex of $S$. A nonempty independent set $S\subseteq V$ is a \emph{small blocking set} of $G$ if
\[
|S|=|N(S)|<\frac{n}{2}.
\]
\end{definition}

We study whether a small blocking set exists and how to construct one when it does. To describe the cardinality relationship between an independent set and its neighborhood, define the \emph{difference} of an independent set $X\subseteq V$ by
\[
d(X)=|X|-|N(X)|.
\]
The small blocking set detection and construction problem can therefore be stated equivalently as finding an independent set $S\subseteq V$ such that
\[
S\neq\varnothing,\qquad d(S)=0,\qquad |S|<\frac{n}{2},
\]
or determining that no such set exists.

To solve this problem, we use critical independent sets.

\begin{definition}[Critical independent set~\cite{larson2007}]
\label{def:critical}
An independent set $I\subseteq V$ is a \emph{critical independent set} (CIS) if
\[
d(I)\geq d(X)
\]
for every independent set $X\subseteq V$. A critical independent set of largest cardinality among all critical independent sets is called a \emph{maximum-cardinality critical independent set}.
\end{definition}

Accordingly, define the critical independence difference of $G$ as
\[
\idc(G)=\max\{d(X):X\subseteq V,\ X\text{ is independent}\}.
\]
Every critical independent set $I$ satisfies $d(I)=\idc(G)$. The maximization allows the empty set, so $d(\varnothing)=0$ implies $\idc(G)\geq 0$.

Larson gave a polynomial-time algorithm for computing a maximum-cardinality critical independent set~\cite{larson2007}. We use the independent set $I$ returned by this algorithm as the starting point for a case analysis based on its difference and cardinality, leading to a polynomial-time algorithm for small blocking set detection and construction.

\section{Classification using critical independent sets}
\label{sec:classification}

Compute a maximum-cardinality critical independent set $I$ and let
\[
\delta=d(I)=\idc(G).
\]
Since $\delta\geq 0$, we consider the cases $\delta>0$ and $\delta=0$ separately.

\subsection{Positive critical difference}
\label{subsec:positive}

We use the following result of Ayoubi and Narayanan. Here, a \emph{blocking set} is a nonempty independent set $S$ satisfying $|S|=|N(S)|$, without the strict cardinality requirement in Definition~\ref{def:small}.

\begin{lemma}[Ayoubi and Narayanan~{\cite[Lemma 5]{ayoubi2025}}]
\label{lem:deletion}
If a nonempty independent set $I$ in a connected graph satisfies
\[
|I|\geq |N(I)|,
\]
then deleting vertices one at a time constructs a blocking set $S\subseteq I$ satisfying
\[
|S|\leq |N(I)|.
\]
\end{lemma}

This immediately gives the required positive-difference branch.

\begin{corollary}
\label{cor:positive}
If $\delta>0$, a small blocking set can be constructed in polynomial time.
\end{corollary}

\begin{proof}
Since $\delta>0$, we have $|I|>|N(I)|$. By Lemma~\ref{lem:deletion}, we can construct a blocking set $S\subseteq I$ with $|S|\leq |N(I)|$. Since $I$ is independent,
\[
|I|+|N(I)|\leq n.
\]
Together with $|I|>|N(I)|$, this implies $2|N(I)|<n$. Hence
\[
0<|S|=|N(S)|\leq |N(I)|<\frac{n}{2}.
\]
Thus, $S$ is a small blocking set.
\end{proof}

\subsection{Zero critical difference}
\label{subsec:zero}

\begin{theorem}
\label{thm:zero}
Suppose that $\delta=0$ and $I$ is a maximum-cardinality CIS. Then:
\begin{enumerate}
\item If $I=\varnothing$, the graph contains no nonempty blocking set.
\item If $0<|I|<n/2$, then $I$ itself is a small blocking set.
\item The only remaining case is $|I|=n/2$.
\end{enumerate}
\end{theorem}

\begin{proof}
Suppose first that $I=\varnothing$ and, for a contradiction, that there is a nonempty blocking set $S$. Then
\[
d(S)=0=\idc(G),
\]
so $S$ is a nonempty CIS, contradicting the choice of the empty set as a maximum-cardinality CIS.

If $0<|I|<n/2$, then $d(I)=0$ gives $|I|=|N(I)|$, so $I$ itself is a small blocking set.

Finally, since $I$ is independent and $d(I)=0$,
\[
2|I|=|I|+|N(I)|\leq n.
\]
Thus, $|I|\leq n/2$. After excluding the empty set and the case $|I|<n/2$, only $|I|=n/2$ remains.
\end{proof}

The maximum-cardinality requirement is used here primarily to justify the empty-set branch. If an arbitrary CIS is computed instead, the subroutine may return the empty set even when a nonempty blocking set exists.

\section{Boundary structure and the projection theorem}
\label{sec:projection}

\subsection{Boundary structure}
\label{subsec:boundary}

By Theorem~\ref{thm:zero}, the only case requiring further analysis is
\[
\delta=0,\qquad |I|=\frac{n}{2}.
\]
Let $J=V\setminus I$.

\begin{theorem}
\label{thm:matching}
If $\delta=0$ and $|I|=n/2$, then
\[
N(I)=J,\qquad |I|=|J|,
\]
and the bipartite graph $H=(I,J;E_H)$ has a perfect matching, where
\[
E_H=\{xy\in E:x\in I,\ y\in J\}.
\]
\end{theorem}

\begin{proof}
Since $d(I)=0$, we have $|N(I)|=|I|=n/2$. Independence of $I$ implies that $I$ and $N(I)$ are disjoint, and their cardinalities sum to $n$. Hence
\[
N(I)=V\setminus I=J.
\]
For any $X\subseteq I$, independence of $X$ and $\idc(G)=0$ imply $|N(X)|\geq |X|$. Moreover, all neighbors of $X$ lie in $J$ because $I$ is independent. Therefore $N_H(X)=N(X)$. Thus, $H$ satisfies Hall's condition and has a matching saturating $I$. Since $|I|=|J|$, this matching is perfect.
\end{proof}

\subsection{Projection of small blocking sets}
\label{subsec:projection}

\begin{theorem}[Projection theorem]
\label{thm:projection}
Let $I$ be any maximum-cardinality critical independent set of $G$ such that
\[
\idc(G)=0,\qquad |I|=\frac{n}{2}.
\]
Then $G$ contains a small blocking set if and only if $I$ contains a nonempty proper subset $T\subsetneq I$ satisfying
\[
|T|=|N(T)|.
\]
Such a subset is called a \emph{tight subset} of $I$.
\end{theorem}

\begin{proof}
Let $J=V\setminus I$. By Theorem~\ref{thm:matching}, the bipartite graph consisting of the edges between $I$ and $J$ has a perfect matching. Fix any such matching $M$.

We first prove the forward implication. Let $S$ be any small blocking set and define
\[
A=S\cap I,\qquad B=S\cap J.
\]
Then $S=A\cup B$ and $A\cap B=\varnothing$. Since $S$ is independent, there are no edges between $A$ and $B$. Write
\[
N_I(B)=N(B)\cap I.
\]
The perfect matching assigns a distinct matched vertex in $J$ to every vertex of $A$, so $|N(A)|\geq |A|$. None of these matched vertices belongs to $B$, since otherwise $S$ would contain an edge. Similarly, each vertex of $B$ has a distinct matched vertex in $I$, giving $|N_I(B)|\geq |B|$.

Because $A\subseteq I$ and $I$ is independent, $N(A)\subseteq J$. On the other hand, $N_I(B)\subseteq I$. These two neighborhood sets are therefore disjoint, and both are subsets of $N(S)$. Consequently,
\[
|N(S)|\geq |N(A)|+|N_I(B)|\geq |A|+|B|=|S|.
\]
Since $S$ is a blocking set, $|N(S)|=|S|$, so equality holds throughout this chain. In particular,
\[
|N(A)|=|A|,\qquad |N_I(B)|=|B|.
\]
We now distinguish whether $S\cap I$ is empty.

\medskip
\noindent\textbf{Case 1: $S\cap I\neq\varnothing$.}
Set $T=A=S\cap I$. Since $T\subseteq I$, the set $T$ is independent. The equality above gives $|N(T)|=|T|$. Furthermore,
\[
0<|T|\leq |S|<\frac{n}{2}=|I|,
\]
so $\varnothing\neq T\subsetneq I$. Thus, $T$ is a nonempty proper tight subset of $I$.

\medskip
\noindent\textbf{Case 2: $S\cap I=\varnothing$.}
Here $S=B\subseteq J$, and the equality above gives $|N_I(S)|=|S|$. Since
\[
N_I(S)\subseteq N(S),\qquad |N(S)|=|S|,
\]
we obtain $N_I(S)=N(S)$ and hence $N(S)\subseteq I$. Let
\[
C=N(S),\qquad T=I\setminus C.
\]
Since $0<|C|=|S|<|I|$, we have
\[
0<|T|=|I|-|C|<|I|.
\]
We claim that $N(T)=J\setminus S$.

First, $T\subseteq I$ and independence of $I$ imply $N(T)\subseteq J$. If some $t\in T$ were adjacent to a vertex $s\in S$, then $t\in N(S)=C$, contradicting $T=I\setminus C$. Therefore,
\[
N(T)\subseteq J\setminus S.
\]
Conversely, the perfect matching matches each vertex of $S$ to a distinct vertex of $C$. Since $|C|=|S|$, these matched vertices exhaust $C$. The remaining matching edges therefore match $J\setminus S$ bijectively to $I\setminus C=T$. Hence
\[
J\setminus S\subseteq N(T).
\]
Combining the two inclusions yields $N(T)=J\setminus S$. It follows that
\[
|N(T)|=|J|-|S|,\qquad
|T|=|I|-|C|=|I|-|S|.
\]
Since $|I|=|J|$, we have $|N(T)|=|T|$. Thus, $T$ is again a nonempty proper tight subset of $I$, proving the forward implication.

For the reverse implication, suppose that $\varnothing\neq T\subsetneq I$ and $|T|=|N(T)|$. Since $T\subseteq I$, it is independent, and
\[
0<|T|<|I|=\frac{n}{2}.
\]
Thus, $T$ itself is a small blocking set.
\end{proof}

Theorem~\ref{thm:projection} shows that, in the case $\delta=0$ and $|I|=n/2$, any graph containing a small blocking set also contains one within $I$. It therefore suffices to detect and construct such a set within $I$.

\section{Algorithm design and analysis}
\label{sec:algorithm}

\subsection{Constructing a small blocking set within a maximum-cardinality critical independent set}
\label{subsec:digraph}

The projection theorem reduces the problem on the entire graph to deciding whether $I$ contains a nonempty proper tight subset. We perform this test using the standard matching-contraction construction.

\begin{theorem}
\label{thm:scc}
Let $I$ be any maximum-cardinality critical independent set of $G$ such that
\[
\idc(G)=0,\qquad |I|=\frac{n}{2}.
\]
Choose any perfect matching in $H=(I,J;E_H)$,
\[
M=\{x_i y_i:1\leq i\leq k\},
\]
where
\[
I=\{x_1,\ldots,x_k\},\quad
J=\{y_1,\ldots,y_k\},\quad k=\frac{n}{2}.
\]
Construct a digraph $D_M$ with vertex set $[k]=\{1,\ldots,k\}$. For distinct $i,j\in[k]$, add an arc $i\to j$ if $x_i y_j\in E$. Then
\[
G\text{ contains a small blocking set}
\quad\Longleftrightarrow\quad
D_M\text{ is not strongly connected}.
\]
When $D_M$ is not strongly connected, every sink strongly connected component $C$ gives a small blocking set of $G$,
\[
T_C=\{x_i:i\in C\}.
\]
\end{theorem}

\begin{proof}
By Theorem~\ref{thm:projection}, in the case $\idc(G)=0$ and $|I|=n/2$, the graph $G$ contains a small blocking set if and only if $I$ contains a nonempty proper tight subset. If $k=1$, then $I$ has no nonempty proper subset and $D_M$ is a single-vertex strongly connected digraph, so the assertion holds directly. We henceforth consider $k\geq 2$.

For a balanced bipartite graph with a perfect matching, classical Dulmage--Mendelsohn decomposition theory states that a nonempty proper subset $S\subsetneq I$ satisfying $|N(S)|=|S|$ exists if and only if the auxiliary digraph induced by the perfect matching is not strongly connected~\cite{berczi2018}. For any chosen perfect matching $M$, our digraph $D_M$ is precisely the corresponding directed representation obtained by orienting edges from $I$ to $J$, contracting each matched pair $x_i y_i$, and omitting self-loops. Therefore,
\[
I\text{ contains a nonempty proper tight subset}
\quad\Longleftrightarrow\quad
D_M\text{ is not strongly connected}.
\]
Here, all neighborhoods of subsets of $I$ are the same in $H$ and $G$.

This test does not depend on the particular choice of $M$. Indeed, for any nonempty proper tight subset $S\subsetneq I$ and any perfect matching $M$, the set $M(S)$ of vertices matched to $S$ satisfies
\[
M(S)\subseteq N(S),\qquad |M(S)|=|S|=|N(S)|,
\]
and hence $M(S)=N(S)$. Thus, under any perfect matching, $S$ corresponds to a nonempty proper vertex subset of $D_M$ with no arcs leaving it. The choice of perfect matching cannot change the outcome of the strong-connectivity test. Combining this equivalence with the projection theorem gives
\[
G\text{ contains a small blocking set}
\quad\Longleftrightarrow\quad
D_M\text{ is not strongly connected}.
\]

When $D_M$ is not strongly connected, its condensation digraph is a directed acyclic graph (DAG) with at least two vertices, and hence has a sink SCC $C$. As $C$ has no arcs to vertices outside it, consider $T_C=\{x_i:i\in C\}$. Since $T_C\subseteq I$ and $I$ is independent, every neighbor of $T_C$ lies in $J$. The matching edges guarantee
\[
\{y_i:i\in C\}\subseteq N(T_C),
\]
while the absence of arcs leaving $C$ guarantees that $T_C$ is not adjacent to any $y_j$ with $j\notin C$. Therefore,
\[
N(T_C)=\{y_i:i\in C\},
\]
so $|T_C|=|N(T_C)|$. Moreover,
\[
0<|T_C|<|I|=\frac{n}{2}.
\]
Since $T_C\subseteq I$ is independent, it is a small blocking set.

If $D_M$ is strongly connected, its only SCC corresponds to the entire set $I$, whose cardinality is $n/2$, so it cannot be output as a small blocking set. The equivalence above requires the algorithm to return \texttt{NO} in this case.
\end{proof}

\subsection{Algorithmic framework}
\label{subsec:framework}

\noindent\textbf{Algorithm 1: Small blocking set detection and construction.}

\smallskip
\noindent\textbf{Input:} A connected simple graph $G=(V,E)$.

\noindent\textbf{Output:} A small blocking set, or \texttt{NO}.

\begin{enumerate}
\item If $n\leq 1$, return \texttt{NO}.
\item Compute a maximum-cardinality CIS $I$ using Larson's algorithm, and set $\delta=|I|-|N(I)|$.
\item If $\delta>0$, apply the vertex-deletion construction of Lemma~\ref{lem:deletion} and return the resulting small blocking set $S$.
\item If $\delta=0$ and $I=\varnothing$, return \texttt{NO}.
\item If $\delta=0$ and $0<|I|<n/2$, return $I$.
\item In the remaining case, $\delta=0$ and $|I|=n/2$. Let $J=V\setminus I$, construct the bipartite graph $H$, and compute any perfect matching $M$.
\item Construct $D_M$ and compute its strongly connected components.
\item If $D_M$ is strongly connected, return \texttt{NO}. Otherwise, choose any SCC $C$ whose vertex in the condensation digraph has outdegree zero, and return $T_C=\{x_i:i\in C\}$.
\end{enumerate}

\subsection{Correctness}
\label{subsec:correctness}

\begin{theorem}
\label{thm:correctness}
Algorithm 1 correctly decides whether a connected simple graph contains a small blocking set and constructs one whenever it exists.
\end{theorem}

\begin{proof}
When $n\leq 1$, no nonempty set can satisfy the required cardinality bound, so step 1 is correct. Suppose henceforth that $n\geq 2$.

If $\delta>0$, Corollary~\ref{cor:positive} guarantees that step 3 outputs a small blocking set.

If $\delta=0$ and $I=\varnothing$, Theorem~\ref{thm:zero} guarantees that no nonempty blocking set exists, so step 4 is correct.

If $\delta=0$ and $0<|I|<n/2$, Theorem~\ref{thm:zero} guarantees that $I$ itself is a small blocking set, so step 5 is correct.

The only remaining case is $\delta=0$ and $|I|=n/2$. By Theorem~\ref{thm:matching}, the perfect matching required in step 6 exists. By Theorems~\ref{thm:projection} and~\ref{thm:scc}, the graph contains a small blocking set if and only if $D_M$ is not strongly connected; in that case, the set corresponding to the chosen sink SCC is a valid small blocking set. Thus, every branch is correct.
\end{proof}

\subsection{Running time}
\label{subsec:complexity}

We compute a maximum-cardinality critical independent set using Larson's MCIS algorithm~\cite{larson2007}, with the matching-reuse implementation described by Strash~\cite{strash2016}. For connected graphs, this step takes $O(m^2)$ time.

When $\delta>0$, we use the iterative deletion procedure in Lemma~\ref{lem:deletion}, removing a vertex from the current independent set and updating its neighborhood at each iteration. With adjacency lists and counts of neighbors in the current independent set, each vertex is deleted at most once and each edge is processed only a constant number of times. This procedure therefore takes $O(n+m)$ time.

In the boundary case, the Hopcroft--Karp algorithm~\cite{hopcroft1973} computes a perfect matching of $H$ in $O(m\sqrt{n})$ time. Constructing $D_M$, computing its SCCs with Tarjan's algorithm~\cite{tarjan1972}, and selecting a sink SCC all take linear time, namely $O(n+m)$.

Since $G$ is connected, $m\geq n-1$. Consequently, the overall running time is $O(m^2)$.

\section{Conclusions and future work}
\label{sec:conclusion}

We studied the detection and construction of small blocking sets in connected simple graphs. A maximum-cardinality critical independent set provides a case classification, and, in the balanced boundary case, the projection theorem reduces the search over the entire graph to the detection of a tight subset within the selected independent set. Matching contraction then reduces this problem to a strong-connectivity test in an auxiliary digraph, and a sink SCC yields a small blocking set. The resulting algorithm runs in $O(m^2)$ time. Future work may investigate base-graph repair by adding the minimum number of edges needed to eliminate all small blocking sets, providing low-cost structural modifications that guarantee universal restless explorability for sufficiently long lifetimes. Another direction is to decide the feasibility of restless exploration and construct exploration schedules for a given permutation sequence, allowing base graphs that do not satisfy the universal explorability condition to be explored under suitable temporal arrangements.


\begin{thebibliography}{00}

\bibitem{holme2012}
P. Holme, J. Saram\"aki,
Temporal networks,
Physics Reports 519 (3) (2012) 97--125.

\bibitem{casteigts2012}
A. Casteigts, P. Flocchini, W. Quattrociocchi, N. Santoro,
Time-varying graphs and dynamic networks,
International Journal of Parallel, Emergent and Distributed Systems 27 (5) (2012) 387--408.

\bibitem{pan2011}
R. K. Pan, J. Saram\"aki,
Path lengths, correlations, and centrality in temporal networks,
Physical Review E 84 (1) (2011) 016105.

\bibitem{salathe2010}
M. Salath\'e, M. Kazandjieva, J. W. Lee, P. Levis, M. W. Feldman, J. H. Jones,
A high-resolution human contact network for infectious disease transmission,
Proceedings of the National Academy of Sciences 107 (51) (2010) 22020--22025.

\bibitem{kuhn2010}
F. Kuhn, N. Lynch, R. Oshman,
Distributed computation in dynamic networks,
in: Proceedings of the 42nd ACM Symposium on Theory of Computing (STOC 2010), 2010, pp. 513--522.

\bibitem{michail2016}
O. Michail, P. G. Spirakis,
Traveling salesman problems in temporal graphs,
Theoretical Computer Science 634 (2016) 1--23.

\bibitem{erlebach2021}
T. Erlebach, M. Hoffmann, F. Kammer,
On temporal graph exploration,
Journal of Computer and System Sciences 119 (2021) 1--18.

\bibitem{erlebach2018}
T. Erlebach, J. T. Spooner,
Faster exploration of degree-bounded temporal graphs,
in: 43rd International Symposium on Mathematical Foundations of Computer Science (MFCS 2018),
Leibniz International Proceedings in Informatics, vol. 117, 2018, pp. 36:1--36:13.

\bibitem{bellitto2023}
T. Bellitto, C. Conchon-Kerjan, B. Escoffier,
Restless exploration of periodic temporal graphs,
in: 2nd Symposium on Algorithmic Foundations of Dynamic Networks (SAND 2023),
Leibniz International Proceedings in Informatics, vol. 257, 2023, pp. 13:1--13:15.

\bibitem{ayoubi2025}
K. Ayoubi, L. Narayanan,
Restless exploration and token dissemination in vertex-permuted temporal graphs,
in: 4th Symposium on Algorithmic Foundations of Dynamic Networks (SAND 2025),
Leibniz International Proceedings in Informatics, vol. 330, 2025, pp. 12:1--12:17.

\bibitem{larson2007}
C. E. Larson,
A note on critical independence reductions,
Bulletin of the Institute of Combinatorics and its Applications 51 (2007) 34--46.

\bibitem{tarjan1972}
R. Tarjan,
Depth-first search and linear graph algorithms,
SIAM Journal on Computing 1 (2) (1972) 146--160.

\bibitem{berczi2018}
K. B\'erczi, S. Iwata, J. Kato, Y. Yamaguchi,
Making bipartite graphs DM-irreducible,
SIAM Journal on Discrete Mathematics 32 (1) (2018) 560--590.

\bibitem{strash2016}
D. Strash,
On the power of simple reductions for the maximum independent set problem,
in: Computing and Combinatorics (COCOON 2016),
Lecture Notes in Computer Science, vol. 9797, 2016, pp. 345--356.

\bibitem{hopcroft1973}
J. E. Hopcroft, R. M. Karp,
An $n^{5/2}$ algorithm for maximum matchings in bipartite graphs,
SIAM Journal on Computing 2 (4) (1973) 225--231.

\end{thebibliography}
\end{document}